\documentclass{article}
\usepackage[utf8]{inputenc}
\usepackage{graphicx}
\graphicspath{ {images/} }
\usepackage{caption}
\usepackage{subcaption}
\usepackage{float}
\usepackage{bm}
\usepackage{tikz}
\usepackage{geometry}
\usepackage{fancyhdr}
\usepackage{setspace}
\usepackage{graphicx}
\usepackage{hyperref}
\hypersetup{colorlinks,allcolors=blue}
\usepackage{amsthm,amsmath}
\usepackage{mathtools}
\usepackage{enumerate}
\usepackage{amsfonts}
\usepackage{bbold}
\usepackage{xcolor}
\usepackage{physics}
\usepackage{booktabs}
\usetikzlibrary{calc}
\newtheorem{theorem}{Theorem}
\newtheorem{lemma}[theorem]{Lemma}

\newtheorem{claim}[theorem]{Claim}
\newtheorem{corollary}[theorem]{Corollary}

\theoremstyle{definition}
\newtheorem{definition}[theorem]{Definition}
\newtheorem{rmk}[theorem]{Remark}

\numberwithin{equation}{section}  
\numberwithin{theorem}{section} 

\newcommand{\ep}{\varepsilon}

\newcommand{\ro}{\varrho}
\newcommand{\brho}{\bar{\varrho}}

\newcommand{\htheta}{\hat{\theta}}

\newcommand{\R}{\mathbb{R}}

\title{Fast state transfer via loop weights}
\author{Gabor Lippner\footnote{Department of Mathematics, Northeastern University, Boston MA, g.lippner@northeastern.edu}, Yujia Shi\footnote{Department of Physics, Creighton University, Omaha NE, yujiashi@creighton.edu}}

\begin{document}
\maketitle
\begin{abstract}
    We prove that almost-linear-time high-fidelity state transfer is achievable in a quantum spin chain using loop weights at the second and second-to-last nodes. We provide specific parameter values, and using a careful analysis of the eigenvectors we make precise quantitative estimates of the transfer time and strength.
\end{abstract}

\section{Introduction}

Quantum state transfer is an important phenomenon that enables transmission of quantum information over physical distances via networks of spin particles. The simplest such network is the spin chain, represented by the path graph. The study of state transfer on it has been initiated by Bose~\cite{bose2003} and has since received a lot of attention from both the physics and mathematics communities. We refer the reader to the surveys~\cite{kay2010, QST} for further literature from both of these perspectives.

The most important features of the state transfer are its fidelity (how likely is the information transferred?) and the transfer time (how long do you have to wait for the transfer to happen?). It has been long known that without some inhomogeneity in the chain, the transfer strength cannot typically be arbitrarily close to 1. One such way of introducing inhomogeneity is to apply magnetic fields at certain positions of the chain. This has been explored extensively, for example in~\cite{christandl2005,stolze2012,casaccino2009,path}, for the case of the magnetic field being applied to the first and last particles of the chain. It turns out that one can indeed achieve $1-\ep$ fidelity state transfer using magnetic fields of strength $O(1/\ep)$, but the transfer time will become exponentially large in the length of the chain - a side-effect that makes this approach impractical. 

However, Chen et al.~\cite{chen2016} observed a rather unexpected phenomenon: high fidelity state transfer can be achieved in polynomial time by applying inhomogeneities near, but not exactly at, the two ends.  
Yet, \cite{chen2016} has a shortcoming: they place the inhomogeneity at the 3rd nodes of the chain. This results in the analysis becoming too complex to rigorously carry out, hence they fall back to numerical evidence coupled with a heuristic argument. More importantly, their protocol results in a highly sensitive readout time in the sense that the transfer fidelity fluctuates rapidly near the optimal time. See Figure~\ref{fig:feder}.

\begin{figure}
\begin{center}
\begin{minipage}{0.35\textwidth}
    \centering
    \includegraphics[width=\textwidth]{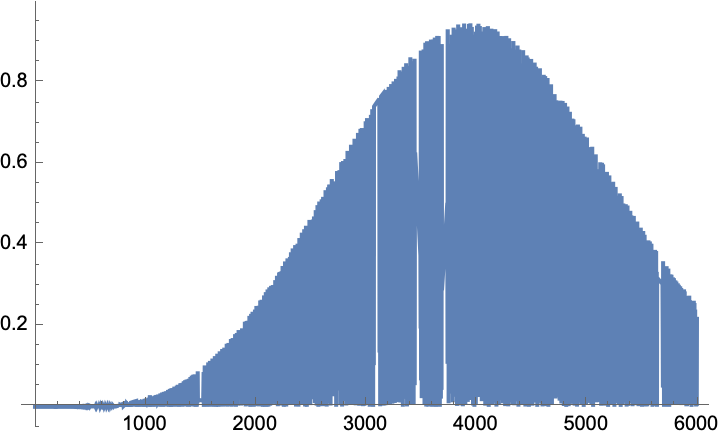}
\end{minipage}
\begin{minipage}{0.35\textwidth}
    \centering
    \includegraphics[width=\textwidth]{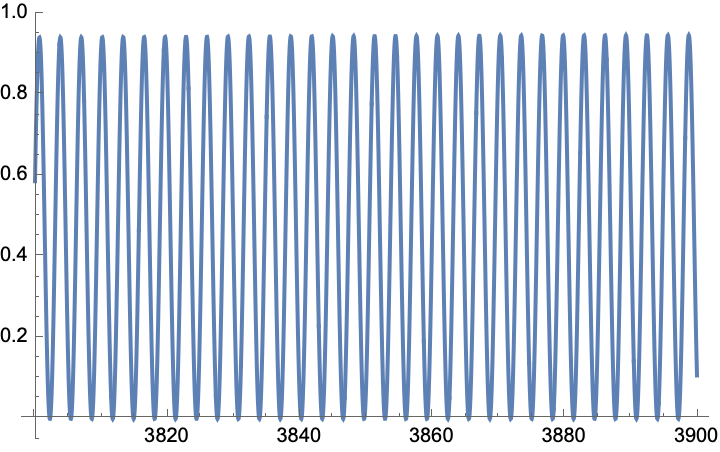}
\end{minipage}
    \caption{\label{fig:feder} Transfer fidelity as a function of time in the Chen et al. protocol for a chain of length $N=501$. The second image is zoomed in to near the optimal transfer time.}
\end{center}
\end{figure}

Our goal here is to present a simplification of their construction that addresses all the aforementioned issues: we are able to provide completely rigorous mathematical analysis, remove the sensitivity of the readout time (see Figure~\ref{fig:us}), and even allow for calibration of the fidelity--time tradeoff.

\subsection{Results}\label{sec:results}

We consider a quantum $XX$ spin chain on $n$ nodes with uniform couplings, and apply a magnetic field of strength $Q$ at the 2nd and 2nd-to-last nodes. Then the Hamiltonian restricted to the 1-excitation subspace becomes 
\[
H = A_{\text{path}} +
Q \cdot D_{2,n-1}
\] where $D_{2,n-1}$ denotes the diagonal matrix whose entries in the 2nd and $n-1$st rows are 1, the rest are 0. 

\begin{definition}
    The continuous time quantum walk of the system is given by $U(t) = e^{it H}$. If $|U(t_0)_{1,n}| > 1- \ep$, we say that there is quantum state transfer between the ends of the chain at time $t_0$ with fidelity at least $1-\ep$. 
\end{definition}

\begin{theorem}\label{thm:main}
Fix $\ep > 0$. Then for $n = \Omega(1/\ep)$ there is a $Q = O(\sqrt{n}/\ep)$ such that we get quantum state transfer between the chains ends with fidelity at least $1-\ep$ and in time $t_0 < O(n/\ep)$.
\end{theorem}

In order to prove this result, in Section~\ref{sec:vectors} we describe the exact form of the eigenvectors of $H$. Then, in Section~\ref{sec:fidelity} we will show that for suitably chosen $Q$, exactly two of the eigenvectors will have most of their weight supported on nodes 1 and $n$. Thus by the symmetry of the system they will be approximately $(1/\sqrt{2}, 0, 0, \dots, 0, 1/\sqrt{2})$ and $(1/\sqrt{2}, 0, 0, \dots, 0, -1/\sqrt{2})$. If the corresponding eigenvalues are denoted by $\lambda_1$ and $\lambda_2$, then we get strong state transfer at time $t_0 = \pi/(\lambda_1-\lambda_2)$. 

\begin{figure}[h!]
\begin{center}
\begin{minipage}{0.35\textwidth}
    \centering
    \includegraphics[width=\textwidth]{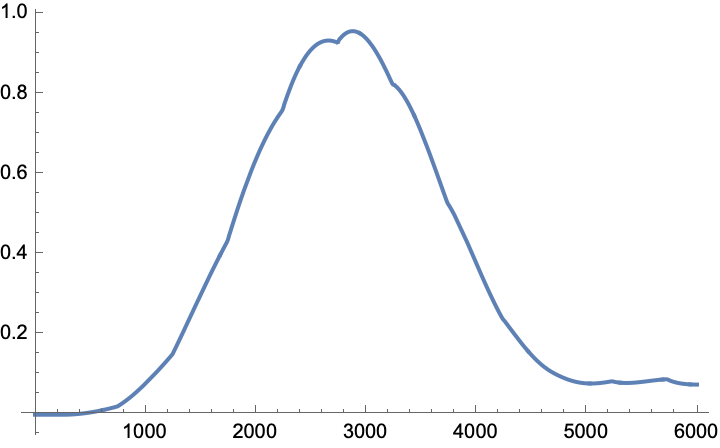}
\end{minipage}
\begin{minipage}{0.35\textwidth}
    \centering
    \includegraphics[width=\textwidth]{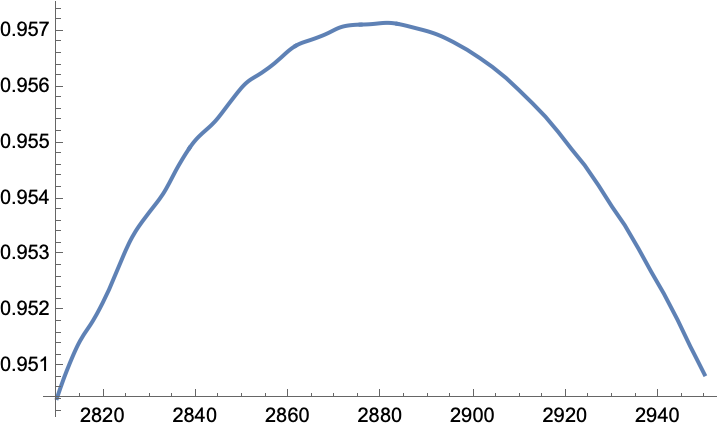}
\end{minipage}
    \caption{\label{fig:us} Transfer fidelity as a function of time in our protocol for a chain of length $N=501$ and $Q = 80$ The second image is zoomed in to near the optimal transfer time. Note that our protocol has transfer fidelity consistently above 0.95 in a large time window.}
\end{center}
\end{figure}

\section{Eigenvectors} \label{sec:vectors}

We will assume throughout that $Q \gg 1$.
Then, according to the Gershgorin Circle Theorem, exactly two eigenvalues of $H$ will fall within $Q\pm 2$, while all the others will fall between $\pm 2$. Since the eigenvectors corresponding to the two larger eigenvalues will be both concentrated on the 2nd and $n-1$st nodes, these won't play a significant role in the state transfer we are looking for. Hence we will be primarily concerned with the ``small'' eigenvalues.  
Due to the reflection symmetry of the system, we can choose an eigenbasis so that any vector \( \varphi \) in the basis satisfies \( \forall j: \varphi(j) = \varphi(n+1-j) \) or \( \forall j: \varphi(j) - \varphi(n+1-j) \). We call these symmetric (respectively, alternating) eigenvectors.

\subsection{Symmetric eigenvectors}

First, we investigate the symmetric case in detail.  
Let \( \lambda\) denote an eigenvalue in $(-2,2)$ whose corresponding eigenvector \( \varphi \) is symmetric with $\varphi(1)=\varphi(n)=1$. Let us denote the remaining entries as follows:
\[
\varphi = (1, a_0, a_1, \dots, a_k, 1)
\] where $k=n-3$. It turns out $k$ will be, in terms of notation, a more convenient parameter to use than $n$.

As $\varphi$ is assumed to be an eigenvector with eigenvalue $\lambda$, writing out the \( H \cdot \varphi = \lambda \cdot \varphi\) equation coordinate wise, we get a system of equations:
\begin{eqnarray}
\label{eq:a0} a_0 &=& \lambda\\
\label{eq:a0Q} 1+Q\cdot a_0 + a_1 &=& \lambda \cdot a_0\\
a_0 + a_2 &=& \lambda \cdot a_1 \\
a_1 + a_3 &=& \lambda \cdot a_2 \\
&\vdots&
\end{eqnarray}

Starting from the 3rd line these simply form a second-order linear recurrence relation from which it follows that 
\[
 a_j = C\ro^j + C' \brho^j \;\;(j = 0,\dots,k)
\]
where $\ro$ is the root of the characteristic equation $1+x^2 = \lambda x$. Since $|\lambda| < 2$, the two roots $\frac{\lambda \pm \sqrt{\lambda^2 - 4}}{2}$ are distinct complex numbers that are each others' conjugates:
\begin{eqnarray*}
    \ro &=& \frac{\lambda}{2} - \frac{i}{2}\sqrt{4 - \lambda^2}\\
    \brho &=& \frac{\lambda}{2} + \frac{i}{2}\sqrt{4 - \lambda^2}.
\end{eqnarray*} 
It is easy to see that $|\ro| = 1$. Since all the $a_j$s are real, it follows that $C'= \bar C$. Our next step is to gain more information about $C$ from the symmetry condition of the eigenvector. This can be written as $a_j = a_{k-j}$, hence 
\[
C \ro^j + \bar{C} \brho^j = C \ro^{k-j} + \bar{C} \brho^{k-j},
\]
or equivalently
\[
C (\ro^j - \ro^{k-j}) + \bar{C} (\brho^j - \brho^{k-j}) = 0.
\]
This implies that \( C (\ro^j - \ro^{k-j}) \) and its conjugate \( \bar{C} (\brho^j - \brho^{k-j}) \) are purely imaginary.

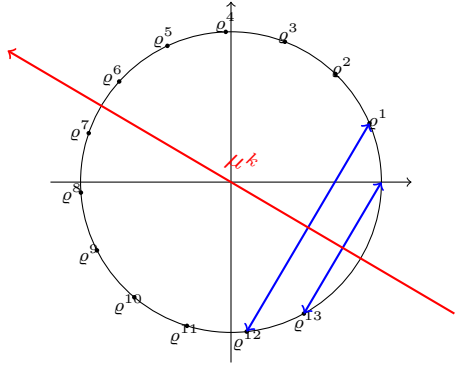
\begin{figure}
\centering
\begin{tikzpicture}[scale=2]

  \draw[->] (-1.2,0) -- (1.2,0);
  \draw[->] (0,-1.2) -- (0,1.2);

  \draw (0,0) circle (1);

  \foreach \n in {1,...,13} {
    \pgfmathsetmacro{\ang}{\n*23}  
    \coordinate (P\n) at ({cos(\ang)},{sin(\ang)});
    \fill (P\n) circle (0.015);
    \node[
      font=\scriptsize,
      shift={({0.12*cos(\ang)},{0.12*sin(\ang)})}
    ] at (P\n) {$\ro^{\n}$};
  }
\draw[<->, blue, thick] (P1) -- (P12);
\draw[<->, blue, thick] (1,0) -- (P13);

\coordinate (M) at ($(1,0)!0.5!(P13)$);
\draw[<-, red, thick]
    ($(0,0)!-2!(M)$) -- ($(0,0)!2!(M)$)
    node[midway, above, sloped] {$\mu^k$};

\end{tikzpicture}
\caption{$\ro^j - \ro^{k-j}$ vs $\mu^k$\label{pic:unit circle}}
\end{figure}

\begin{claim}
Let \(\mu\) be a unit complex number such that \(\mu^2=\ro\).  Then
\[
C = r\cdot\bar{\mu}^{k}
\]
for some real number \(r\).
\end{claim}





\begin{proof}
Choose \(\mu\) so that \(\mu^2=\ro\).  Then for any \(j\),
\[
\ro^j - \ro^{\,k-j}
= \mu^{2j} - \mu^{2(k-j)}
= \mu^k\bigl(\mu^{2j-k} - \mu^{k-2j}\bigr).
\]
But \(\mu^{2j-k}\) and \(\mu^{k-2j}\) are conjugates on the unit circle, so their difference is purely imaginary.  Hence \(\ro^j - \ro^{\,k-j}\) is a purely imaginary multiple of \(\mu^k\), and in particular all such differences are parallel (see Figure~\ref{pic:unit circle}). Thus $\bar{\mu}^k(\ro^j-\ro^{k-j})$ is purely imaginary.  We have seen that the reflection symmetry condition implies that $C(\ro^j-\ro^{k-j})$ is also purely imaginary. This is only possible if
\[
C = r\cdot\bar{\mu}^{k}
\]
for some $r\in \R$ as claimed.
\end{proof}

Recall that $\ro+\brho = \lambda$, while equation \eqref{eq:a0} implies that $\lambda = a_0 = C + \bar{C}= r(\mu^k+\bar{\mu}^k)$. Hence we get 
\[
r=\frac{\ro+\brho}{\mu^k+ \bar{\mu}^k}
\]
Substituting back into the formula for \(C\) yields
\[  
C = \bar{\mu}^k \cdot \frac{\ro+\brho}{\mu^k+ \bar{\mu}^k}
= \frac{\ro+\brho}{\ro^k+1}.
\]
 Thus, the component of the eigenvector are
\[
a_j = \frac{\ro+\brho}{\ro^k+1} \cdot \ro^j + \frac{\ro+\brho}{\brho^k+1} \cdot \brho^j = (\ro+\brho)\frac{\ro^{j-k/2} + \brho^{j-k/2}}{\ro^{k/2}+\brho^{k/2}}
\]
The final step is to relate $Q$ and $\ro$ which can be done using equation~\eqref{eq:a0Q}.
\[
\lambda\,a_0 = 1 + Q\cdot a_0 + a_1.
\]
From the previous equations this can be written as 
\[
(\ro+\brho)^2
= 1 + Q\,(\ro+\brho)
  + (\ro+\brho)\left(\frac{\ro}{1+\ro^k} + \frac{\brho}{1+\brho^k}\right).
\]
Next, dividing both sides by 
\(\ro+\brho\) yields
\[
\ro+\brho
= \frac{1}{\ro+\brho}
  + Q
  + \frac{\ro}{1+\ro^k}
  + \frac{\brho}{1+\brho^k}.
\]
Rearranging to isolate \(Q\) leads to

\begin{eqnarray} \label{eq:q+1/lambda}
    Q &=& \frac{-1}{\ro+\brho}+
 \ro+\brho
  - \frac{\ro}{1+\ro^k}
  - \frac{\brho}{1+\brho^k}\\&=&
     \frac{-1}{\ro+\brho}
+ \frac{\ro^{k+1}}{1+\ro^k} + \frac{\brho^{k+1}}{1+\brho^k} \\
 &=& 
\frac{-1}{\ro+\brho} + \frac{\ro^{k/2+1}+\brho^{k/2+1}}{\ro^{k/2} + \brho^{k/2}}.
\end{eqnarray}
Finally, let us denote the argument of $\ro$ by $\theta$. Then this last equation can be written in terms of trigonometric functions as follows.

\[
Q = \frac{-1}{2\cos\theta}
+ \frac{\cos(\tfrac{k+2}{2}\theta)}
       {\cos(\tfrac{k}{2} \theta)}= \frac{-1}{2\cos \theta}+\cos \theta - \tan \left(\tfrac{k}{2}\theta\right)\sin \theta .
\]
As a reminder, the eigenvalue is then written as 
\[
\lambda = \ro + \brho
        = 2\cos \theta.
\]

The following lemma summarizes the results of all these computations.

\begin{lemma}\label{lem:sym_case}
    For $Q > 2$, any eigenvalues $\lambda \in (-2,2)$ that corresponds to a symmetric eigenvector can be written as $\lambda = 2\cos \theta$ where $\theta$ is a solution of 
\[    Q =  \frac{-1}{2\cos \theta}+\cos \theta - \tan \left(\tfrac{k}{2}\theta\right)\sin \theta, \] 
and the corresponding eigenvector can be chosen as 
\[ \frac{\cos(\tfrac{k}{2}\theta)}{\lambda}\varphi = \left(\frac{\cos(\tfrac{k}{2}\theta)}{\lambda}, \cos(\tfrac{k}{2}\theta),\cos(\tfrac{k-2}{2}\theta),\cos(\tfrac{k-4}{2}\theta),\dots, \cos(\tfrac{k-2k}{2}\theta),\frac{\cos(\tfrac{k}{2}\theta)}{\lambda}\right).\]
Vice versa, all such $\theta$ solutions give rise to eigenvectors of this type.
\end{lemma}

\begin{rmk}
    Our goal is to arrange $Q$ and $k$ such that one of these eigenvectors has almost all of its weight supported on the first and last nodes. This requires $\lambda$ (and hence $\cos \theta)$ to be very small in absolute value. At the same time $\cos(k\theta/2 )$ should not be too small. Specifically, as we will later establish, we need 
    \[|\lambda| \leq \left(\tfrac{\sqrt{\ep}}{\sqrt{k+1}} |\cos(k\theta/2)|\right)\] 
    in order for this eigenvector to help with $1-\ep$ fidelity state transfer.
\end{rmk}

\begin{rmk}
    For our purposes we don't need to know that all eigenvectors are of the form given above, just that we can find one in this form. Then the proof of the lemma could be shortened to simply verifying that such a vector is indeed an eigenvector for the given $Q$. We decided to include the whole calculation both because it shows that all eigenvectors are of this form and explains how this specific form was found.
\end{rmk}

\subsection{Alternating eigenvectors}

An entirely analogous analysis can be carried out for the alternating case. We state the summarizing lemma first, and then point out the minor differences between the two cases.

\begin{lemma}\label{lem:alt_case}
    For $Q > 2$, any eigenvalue $\lambda \in (-2,2)$ that corresponds to an antisymmetric eigenvector can be written as $\lambda = 2\cos \theta$ where $\theta$ is a solution of 
\[    Q =  \frac{-1}{2\cos \theta}+\cos \theta + \cot \left(\tfrac{k}{2}\theta\right)\sin \theta, \] 
and the corresponding eigenvector can be chosen as 
\[ \frac{\sin(\tfrac{k}{2}\theta)}{\lambda}\varphi = \left(\frac{\sin(\tfrac{k}{2}\theta)}{\lambda}, \sin(\tfrac{k}{2}\theta),\sin(\tfrac{k-2}{2}\theta),\sin(\tfrac{k-4}{2}\theta),\dots, \sin(\tfrac{k-2k}{2}\theta),\frac{\sin(\tfrac{-k}{2}\theta)}{\lambda}\right).\]
Vice versa, all such $\theta$ solutions give rise to eigenvectors of this type.
\end{lemma}

\begin{proof}
    We search for the eigenvector in the same form, leading to the same system of equations as in the symmetric case. The general solution can still be written as $a_j = C \ro^j + \bar{C}\brho^j$. However, the alternating condition now implies that 
    \[ 0 = a_j + a_{k-j} = C(\ro^j+\ro^{k-j}) + \bar{C}(\brho^j + \brho^{k-j}), \] thus 
    $C(\ro^j + \ro^{k-j})$ has to be purely imaginary. Using $\ro = \mu^2$ as before, we can compute \[\ro^j + \ro^{k-j} = \mu^{2j}+\mu^{2k-2j} = \mu^k (\mu^{2j-k} + \mu^{k-2j}) = \mu^k(\mu^{2j-k}+\bar{\mu}^{2j-k}),\] so $\bar{\mu}^k (\ro^j+\brho^j)$ is real, hence $C$ must be of the form 
    \[ C = i r \bar{\mu}^k\] where $r \in \R$. Next we get 
    \[ \ro+\brho = \lambda = a_0 = C + \bar{C} = ir(\bar{\mu}^k-\mu^k)\] so
    \[ir = \frac{\ro+\brho}{\bar{\mu}^k-\mu^k} \] leading to 
    \[ C = ir \bar{\mu}^k = \frac{\ro+\brho}{1-\ro^k}\] and 
    \[ a_j = (\ro+\brho) \left( \frac{\ro^j}{1-\ro^k} + \frac{\brho^j}{1-\brho^k}\right) = \lambda \frac{\ro^{k/2-j}-\brho^{k/2-j}}{\ro^{k/2}-\brho^{k/2}}\]

    Last, but not least, the final equation relating $Q$ and $\ro$ yields
\begin{eqnarray} \label{eq:q+1/lambda:alt}
    Q &=& \frac{-1}{\ro+\brho}+
 \ro+\brho
  - \frac{\ro}{1-\ro^k}
  - \frac{\brho}{1-\brho^k}\\&=&
     \frac{-1}{\ro+\brho}
- \frac{\ro^{k+1}}{1-\ro^k} - \frac{\brho^{k+1}}{1-\brho^k} \\
 &=& 
\frac{-1}{\ro+\brho} + \frac{\ro^{k/2+1}-\brho^{k/2+1}}{\ro^{k/2} - \brho^{k/2}}\\
&=& \frac{-1}{2\cos \theta} + \frac{\sin(\tfrac{k+2}{2} \theta)}{\sin(\tfrac{k}{2}\theta)} = \frac{-1}{2\cos\theta}+\cos \theta + \cot \left(\tfrac{k}{2}\theta\right)\sin \theta,
\end{eqnarray} exactly as claimed.
\end{proof}

\section{Ensuring strong state transfer} \label{sec:fidelity}

In this section we will first establish conditions under which high fidelity state transfer can be guaranteed, and then we will explain how to choose $Q$ and $k$ so that those conditions are satisfied. This involves carefully analyzing the behavior of the two trigonometric functions that appear in Lemmas~\ref{lem:sym_case} and~\ref{lem:alt_case}. We are going to refer to them as
\[ S(\theta) =  \frac{-1}{2\cos \theta}+\cos \theta - \tan \left(\tfrac{k}{2}\theta\right)\sin \theta\] and 
\[ A(\theta) = \frac{-1}{2\cos \theta}+\cos \theta + \cot \left(\tfrac{k}{2}\theta\right)\sin \theta.\]

\subsection{Transfer fidelity}

\begin{lemma}\label{lem:transfer_fidelity}
    Let $\psi_1, \dots, \psi_n$ be an orthonormal eigenbasis for $H$ with corresponding eigenvalues $\lambda_1, \dots, \lambda_n$. Suppose $\psi_1$ is symmetric and $\psi_2$ is alternating, and that both $\psi_1(1)^2, \psi_2(1)^2 \geq \tfrac{1}{2}-\tfrac{\ep}{4}$. Then the transfer fidelity between the endpoints at time $t_0 = \tfrac{\pi}{|\lambda_1-\lambda_2|}$ is at least $1-\ep$.
\end{lemma}

\begin{proof}
 Then, by the usual expansion, the transfer fidelity at time $t_0 = \tfrac{\pi}{|\lambda_1-\lambda_2|}$ can be computed as 
\[ |U(t_0)_{1,n}| = \left|\sum_{j=1}^n e^{it_0 \lambda_j} \psi_j(1) \psi_j(n)\right|.\]
 Then, by the triangle inequality we can estimate
\begin{multline*}
    |U(t_0)_{1,n}| \geq \left|e^{it_0 \lambda_1} \psi_1(1)^2 - e^{it_0 \lambda_2} \psi_2(1)^2\right| - \sum_{j=3}^n \psi_j(1)^2 =\\= \left| e^{it_0 \lambda_1}\left(\psi_1(1)^2 - e^{i t_0 (\lambda_2-\lambda_1)} \psi_2(1)^2\right)\right| - (1-\psi_1(1)^2-\psi_2(1)^2) \\ =2(\psi_1(1)^2)+\psi_2(1)^2)-1.
\end{multline*}
In particular, if both $\psi_1(1)^2, \psi_2(1)^2 \geq \tfrac{1}{2}-\tfrac{\ep}{4}$ then the transfer fidelity at $t_0$ is at least $1-\ep$.
\end{proof}

\begin{claim}\label{cl:psi_bound}
    For $j=1,2$ let $\theta_j$ be the angle in the representation of $\psi_j$ in Lemmas~\ref{lem:sym_case} and~\ref{lem:alt_case}  respectively. If $|\lambda_1| \leq \tfrac{\sqrt{\ep}}{\sqrt{k+1}} |\cos(k\theta_1/2)|$ then $\psi_1(1)^2 \geq \tfrac{1}{2}-\tfrac{\ep}{4}$. Similarly, if $|\lambda_2| \leq \tfrac{\sqrt{\ep}}{\sqrt{k+1}} |\sin(k\theta_2/2)|$ then $\psi_2(1)^2 \geq \tfrac{1}{2}-\tfrac{\ep}{4}$.
\end{claim}

\begin{proof}
     Let $B = \tfrac{\cos(k \theta_1/2)}{\lambda_1}$. By the assumption on $\lambda_1$ we see that $B^2 \geq \tfrac{k+1}{\ep}$. Note that in Lemmas~\ref{lem:sym_case} and~\ref{lem:alt_case} the eigenvectors are not normalized, but the $\psi_j$ eigenvectors in Claim~\ref{cl:psi_bound} are unit length. Thus 
    \[\psi_1(1)^2 = \frac{B^2}{2B^2 + \sum_0^k |\cos(\mbox{something})|^2}\geq \frac{B^2}{2B^2+k+1} \geq \frac{1}{2} \left( 1 - \frac{k+1}{2B^2}\right)\geq \frac{1}{2} - \frac{\ep}{4}\] as claimed. The second part follows in an identical fashion.
\end{proof}

Combining Lemmas~\ref{lem:sym_case},~\ref{lem:alt_case}, and~\ref{lem:transfer_fidelity} with Claim~\ref{cl:psi_bound} leads to the following. 
\begin{corollary}\label{cor:transfer}
  If $\theta_1, \theta_2$ are such that   
  \begin{enumerate}
      \item $S(\theta_1) = Q = A(\theta_2)$
      \item $2|\cos\theta_1| \leq \tfrac{\sqrt{\ep}}{\sqrt{k+1}} |\cos(k\theta_1/2)|$
      \item $2|\cos \theta_2| \leq \tfrac{\sqrt{\ep}}{\sqrt{k+1}} |\sin(k\theta_2/2)|$
  \end{enumerate} then $H$ admits quantum state transfer between the endpoints at time $\tfrac{\pi}{2|\cos \theta_1 - \cos \theta_2|}$ with fidelity at least $1-\ep$.
\end{corollary}

The plan to find such $\theta_1, \theta_2$ is the following. To simplify calculations, we are going to ensure that both $\sin(k\theta_i/2)$ and $\cos(k\theta_i/2)$ are bounded away from 0. Then it is sufficient to pick them such that $|\cos(\theta_i)| \approx \sqrt{\ep/k}$ to satisfy the last two conditions of Corollary~\ref{cor:transfer}.

We are going to achieve this by first fixing $\theta_0$ such that $\cos(k\theta_0/2)=\sin(k\theta_0/2) = 1/\sqrt{2}$ and such that $|\cos(\theta_0)| \approx \sqrt{\ep/k}$. Then we let $Q = F(\theta_0):= \cos\theta_0 - 1/(2\cos \theta_0)$, and show that solutions to $S(\theta_1) = Q = A(\theta_2)$ exist in a small neighborhood of $\theta_0$. To keep $Q$ positive, we are going to choose $\pi/2 < \theta_0 < \pi$ so $\cos(\theta_0) < 0$. 

\begin{theorem}\label{thm:picktheta} Let $0<c<1/2$ be fixed and let $m$ be an integer such that 
\begin{equation}\label{eq:m_cond} 
k+\frac{2c}{\pi}\sqrt{k}+\frac{1}{2} < 8m+1 < k + \frac{1}{\pi}\sqrt{k}-\frac{1}{2}
\end{equation} 
and denote \(\theta_0 = \tfrac{16m+2}{4k}\pi\) and \(Q = F(\theta_0)\). Then there exists $\theta_1, \theta_2$ such that 
    \begin{align*}\frac{16m+1}{4k} \pi < \theta_1 < \frac{16m+2}{4k}\pi < \theta_2 < \frac{16m+3}{4k}\pi  && \mbox{and} && S(\theta_1)= Q = A(\theta_2).\end{align*}
Furthermore, \[\theta_2-\theta_1 \geq \frac{\sqrt{2}-1}{3k}c^2.\]
\end{theorem}

To prove this, we will show that $A(\theta)-Q$ and $S(\theta)-Q$ both change signs in the respective intervals. To do so, we need to investigate the behavior of $S(\theta)$ and $A(\theta)$ in the $(\pi/2,\pi)$ interval.

\subsection{Calculus}

\begin{claim}\label{cl:monotone}
    $F(\theta):= \tfrac{-1}{2\cos \theta} + \cos \theta$ is monotone decreasing and convex on $(\pi/2,\pi)$. Furthermore both 
    \[S(\theta)-F(\theta) = -\tan\left(\tfrac{k}{2} \theta\right) \sin \theta\] and 
    \[A(\theta) - F(\theta)= \cot\left(\tfrac{k}{2} \theta\right) \sin \theta\] are monotone decreasing on $(\pi/2,3\pi/4)$ everywhere they are non-singular.
\end{claim}

\begin{proof} The first statement follows by computing 
\begin{equation}\label{eq:F'} F'(x) = \frac{-\sin x}{2\cos^2 x} - \sin x\end{equation} and 
\[ F''(x) = -\cos x - \frac{1}{2\cos x} - \frac{\sin^2 x}{\cos^3 x} = -\cos x +\frac{1}{2\cos x} - \frac{1}{\cos^3 x},\] and noting that $F'(x) < 0$ term wise on $(\pi/2,\pi)$ and $F''(x) > 0$ as $y+1/y^3 \geq 2/y > 1/(2y)$ for $y = -\cos x >0$ by Cauchy-Schwarz.

     For the second, let us compute the derivative:
    \[\frac{d}{d x} (S(x)-F(x))= -\cos x \tan\left(\tfrac{k}{2} x\right)- \frac{k}{2}\frac{\sin x}{\cos^2\left(\tfrac{k}{2} x\right)}.
    \] To see that this is indeed negative, we need to show 
    \[ k \sin x >-2\cos x \sin\left(\tfrac{k}{2} x\right)\cos\left(\tfrac{k}{2} x\right) = -\cos x \sin(k x). \] This will clearly hold as long as $k\sin x >1$, so at least up to roughly $\pi - 1/k > 3\pi/4$ when $k\geq 2$ while the $k=1$ cases is trivial. The third statement follows in a similar spirit.
\end{proof}

We are going to need a slightly more precise control on $F'(x)$ on an interval near $\pi/2$. Standard trigonometric estimates together with \eqref{eq:F'} imply the following:
\begin{claim}\label{cl:F'}
Let $ 0 < c < 1/2$. Then for any $\tfrac{\pi}{2}+\tfrac{c}{\sqrt{k}} \leq x \leq  \tfrac{\pi}{2} + \tfrac{1}{2\sqrt{k}}$ we have 
\[ k \leq |F'(x)| \leq \frac{3k}{c^2}.\]
\end{claim}

After this preparation we are ready for the 
\begin{proof}[Proof of Theorem~\ref{thm:picktheta}]
Recall that $\theta_0 = \tfrac{8m+1}{2k}\pi$, hence \(\tfrac{k}{2}\theta_0 = 2m\pi + \pi/4\) and thus $\cot(\tfrac{k}{2}\theta_0) = \tan(\tfrac{k}{2}\theta_0)=1$. Then, by the choice of $Q$ we see that 
\[ S(\theta_0 ) = Q -  \sin \theta_0  < Q < Q + \sin(\theta_0) = A(\theta_0).\]
Now let  
\begin{align*} \htheta_1 = \theta_0 - \frac{2}{k}\cdot \frac{\pi}{8} = \frac{16m+1}{4k}\pi && \mbox{and} && \htheta_2 = \theta_0 + \frac{2}{k}\cdot \frac{\pi}{8} = \frac{16m+3}{4k}\pi,\end{align*} and observe that the conditions on $m$ in \eqref{eq:m_cond} imply that 
\begin{align*}\frac{\pi}{2} + \frac{c}{\sqrt{k}} &\leq \htheta_1 & \mbox{and}&& \htheta_2 &\leq \frac{\pi}{2} + \frac{1}{2\sqrt{k}}.
\end{align*} Thus, the conditions of Claim~\ref{cl:F'} are satisfied on the $[\htheta_1,\htheta_2]$ interval. It follows that
\begin{align*}
    F(\htheta_1) &= F(\theta_0)+\int_{\theta_0}^{\htheta_1} F'(x)dx & \mbox{and}& & F(\htheta_2) &= F(\theta_0) + \int_{\theta_0}^{\htheta_2} F'(x)dx \\
    &\geq Q + (\theta_0-\htheta_1)\min_{\htheta_1 \leq x \leq \theta_0} |F'(x)| &&& &\leq Q - (\htheta_2-\theta_0)\min_{\theta_0 \leq x \leq \htheta_2} |F'(x)|\\
    &\geq Q + \frac{\pi}{4k} \cdot k = Q+ \frac{\pi}{4} &&& &\leq Q - \frac{\pi}{4k}\cdot k = Q - \frac{\pi}{4}. \end{align*}
Hence 
\begin{align*}
    S(\htheta_1) &= F(\htheta_1) - \sin \htheta_1 \tan \htheta_1 & \mbox{and}& & A(\htheta_2)&= F(\htheta_2) + \sin \htheta_2 \cot \htheta_2 \\
    &=F(\htheta_1) - \sin \htheta_1 / \sqrt{2}& && &= F(\htheta_2) + \sin \htheta_2/ \sqrt{2} \\
    &\geq F(\htheta_1)-(\sqrt{2}-1) && &&  \leq F(\htheta_2)+\sqrt{2}-1 \\
    &Q + \frac{\pi}{4} - (\sqrt{2}-1) >Q&& && \leq Q - \frac{\pi}{4} + (\sqrt{2}-1) < Q
\end{align*}

Thus, in summary, we see that $S(\htheta_1) > Q > S(\theta_0)$ and $A(\htheta_1) < Q < A(\theta_0)$. Since on the $[\htheta_1, \htheta_2]$ interval both $S$ and $A$ are continuous, there have to exist $\theta_1, \theta_2$ satisfying $\htheta_1 \leq \theta_1 < \theta_0 < \theta_2 \leq \htheta_2$ such that $S(\theta_1) = Q = A(\theta_2)$.

To finish, we need to lower-bound $\theta_2-\theta_1$. We can do that by showing that $F(\theta_1) - F(\theta_2)$ cannot be too small. Using the monotonicity from Claim~\ref{cl:monotone} we can estimate
\begin{align*}F(\theta_1) &= S(\theta_1) + \sin \theta_1 \tan\left(\tfrac{k}{2}\theta_1 \right) 
&\mbox{and}&& F(\theta_2) &= A(\theta_2) - \sin \theta_2 \cot\left(\tfrac{k}{2}\theta_2 \right)
\\ 
&= Q + \sin \theta_1 \tan\left(\tfrac{k}{2}\theta_1 \right) & & &  &= Q - \sin \theta_2 \cot\left(\tfrac{k}{2}\theta_2 \right)
\\ 
&\geq Q + \sin \htheta_1 \tan\left(\tfrac{k}{2}\htheta_1 \right) & & & &\leq Q - \sin \htheta_2 \cot\left(\tfrac{k}{2}\htheta_2 \right)
\\ 
&\geq Q + (\sqrt{2}-1)/2 
&&&   &\leq Q - (\sqrt{2}-1)/2.\end{align*}

So we get, using the upper bound from by Claim~\ref{cl:F'}, that
\[ \sqrt{2}-1 \leq F(\theta_1) - F(\theta_2) \leq (\theta_2 - \theta_1)\max\{|F'(x)|: \theta_1 \leq x \leq \theta_2\} \leq (\theta_2 -\theta_1) \frac{3k}{c^2},\] from which the claimed bound on $\theta_2 - \theta_1$ follows immediately.
\end{proof}

\subsection{Putting it all together}\label{sec:time}

We are ready to prove Theorem~\ref{thm:main}. Let $\ep > 0$ be given. Since $k = n-3 > C/\ep$, we can choose an $m$ such that  
\[ k + \frac{\sqrt{k \ep}}{10}+\frac{1}{2} \leq 8m+1 \leq k + \frac{\sqrt{k \ep}}{9}- \frac{1}{2},\] hence \eqref{eq:m_cond} is satisfied with $c= \sqrt{\ep}\cdot \pi/20$. Then, by Theorem~\ref{thm:picktheta} there exists $\theta_1, \theta_2$ such that 
\[ S(\theta_1) = A
(\theta_2) = Q = F\left(\tfrac{8m+1}{2k}\pi\right),\]
\[ \frac{\pi}{2} < \frac{16m+1}{4k}\pi < \theta_1 < \theta_2 < \frac{16m+3}{4k}\pi < \frac{\pi}{2}+\frac{\sqrt{\ep}}{18 \sqrt{k}}\]
 \[ 2m\pi + \frac{\pi}{8} < \frac{k}{2}\theta_1 < \frac{k}{2}\theta_2 < 2m\pi + \frac{3\pi}{8},\] and 
 \[ \theta_2 - \theta_1 \geq \frac{\sqrt{2}-1}{3k} \frac{\pi^2\ep}{400}\]
Hence both $\cos(\tfrac{k}{2}\theta_1)$ and $\sin(\tfrac{k}{2}\theta_2)$ are at least $\sin(\pi/8) = \tfrac{\sqrt{2-\sqrt{2}}}{2}$.
On the other hand, $\cos(\theta_i): (i=1,2)$ is at most $\tfrac{\pi \sqrt{\ep}}{18\sqrt{k}}$. Hence
\[ 2|\cos(\theta_1)| \leq \frac{\pi\sqrt{\ep}}{9 \sqrt{k}} < \frac{\sqrt{\ep}}{\sqrt{k+1}}\frac{\sqrt{2-\sqrt{2}}}{2} \leq \frac{\sqrt{\ep}}{\sqrt{k+1}}\cos\left(\tfrac{k}{2}\theta_1\right),\] and similarly 
\[ 2|\cos(\theta_2)| \leq \frac{\pi\sqrt{\ep}}{9 \sqrt{k}} < \frac{\sqrt{\ep}}{\sqrt{k+1}}\frac{\sqrt{2-\sqrt{2}}}{2} \leq \frac{\sqrt{\ep}}{\sqrt{k+1}}\sin\left(\tfrac{k}{2}\theta_2\right),\] so all the conditions of Corollary~\ref{cor:transfer} are fulfilled. Thus there is transfer between the endpoints of the chain with fidelity at least $1-\ep$ and in time 
\[ \frac{\pi}{2|\cos(\theta_1)-\cos(\theta_2)|} \approx \frac{\pi}{2|\theta_1 - \theta_2|} \geq \Omega(k/\ep).\]

\paragraph{Acknowledgements:} ~

Gabor Lippner was supported by the Simons Foundation Collaboration Grant \#953420.

\bibliographystyle{plain}
\bibliography{quantum}

\end{document}